\documentclass[review]{elsarticle}

\usepackage{amssymb,amsfonts,amsthm,amsmath}
\usepackage[utf8]{inputenc}
\usepackage{graphicx}
\usepackage{tcolorbox}
\usepackage[T1]{fontenc}
\usepackage{xspace}
\usepackage{todonotes}
\usepackage{bm}
\usepackage{booktabs}
\usepackage{tabularx}
\usepackage{lineno}
\usepackage[ruled,vlined,linesnumbered]{algorithm2e}
\usepackage{mathtools,mathrsfs}
\usepackage{multicol}
\usepackage{pdfpages}
\usepackage{array}
\usepackage{comment}
\usepackage{calc}
\usepackage{caption}
\usepackage{enumerate}
\usepackage{etoolbox}
\usepackage{float}
\usepackage{extarrows}
\tcbuselibrary{skins}
\usepackage{framed}

\usepackage{hyperref}

\usepackage{xcolor}

\colorlet{mix}{red!50!black}

\usepackage{hyperref}
\usepackage{cleveref}

\hypersetup{colorlinks={true},linkcolor={mix},citecolor=mix}

\newcommand{\defparrprob}[4]{
	\vspace{1mm}
	\noindent\fbox{
		\begin{minipage}{0.96\textwidth}
			\begin{tabular*}{\textwidth}{@{\extracolsep{\fill}}lr} #1  & {\bf{}} #3 \\ \end{tabular*}
			{\bf{Input:}} #2  \\
			{\bf{Question:}} #4
		\end{minipage}
	}
	\vspace{1mm}
}

\newcommand{\defparrrprob}[4]{
	\vspace{1mm}
	\noindent\fbox{
		\begin{minipage}{0.96\textwidth}
			\begin{tabular*}{\textwidth}{@{\extracolsep{\fill}}lr} #1   \\ \end{tabular*}
			{\bf{Input:}} #2  \\
			{\bf{Output:}} #3
		\end{minipage}
	}
	\vspace{1mm}
}

\newcommand{\yes}{\textsf{YES}\xspace}
\newcommand{\no}{\textsf{NO}\xspace}

\newcommand{\nph}{\textsf{NP}-hard\xspace}
\newcommand{\fpt}{{\sf FPT}\xspace}

\newcommand{\woh}{{\sf W[1]}-hard\xspace}

\newcommand{\lrtfvd}{{\sc $\ell$-Relaxed Transitive-free Vertex Deletion}\xspace}

\newcommand{\mvdext}{{\sc $\ell$-Relaxed Transitive-free Vertex Deletion-Extension}\xspace}
\newcommand{\mvde}{{\sc   $\ell$-RTVD-Ext}\xspace}

\newcommand{\zvd}{{\sc $0$-RTVD}\xspace}
\newcommand{\mzvd}{{\sc Min  $0$-RTVD}\xspace}
\newcommand{\mvd}{{\sc Min  $\ell$-RTVD}\xspace}
\newcommand{\lvd}{{\sc $\ell$-RTVD}\xspace}
\newcommand{\mlrtvd}{{\sc Minimum $\ell$-Relaxed Transitive-free Vertex Deletion}\xspace}

\newtheorem{observation}{Observation}
\newtheorem{theorem}{Theorem}
\newtheorem{definition}[theorem]{Definition}
\newtheorem{lemma}[theorem]{Lemma}

\newtheorem{corollary}[theorem]{Corollary}
\newtheorem{proposition}[theorem]{Proposition}

\newtheorem{fact}{Fact}
\newtheorem{proc}{Procedure}
\newtheorem{clam}{Claim}

\begin{document}

\begin{frontmatter}
\title{Towards Transitive-free Digraphs}

\author[address1]{Ankit Abhinav}
\ead{ankit.abhinav@niser.ac.in}
\author[address2]{Satyabrata Jana}
\ead{satyamtma@gmail.com}
\author[address1]{Abhishek Sahu}
\ead{abhisheksahu@niser.ac.in}

\address[address1]{National Institute of Science Education and Research, An OCC of Homi Bhabha National Institute,	Bhubaneswar 752050, Odisha, India.}
\address[address2]{University of Warwick, UK}

\begin{abstract}
	In a digraph $D$, an arc $e=(x,y) $ in $D$ is considered transitive if there is a path from $x$ to $y$ in $D- e$. A digraph is transitive-free if it does not contain any transitive arc. In the {\sc Transitive-free Vertex Deletion} (TVD) problem, the goal is to find at most $k$ vertices $S$ such that $D-S$ has no transitive arcs. In our work, we study a more general version of the TVD problem, denoted by {\sc $\ell$-Relaxed Transitive-free Vertex Deletion} ({\sc $\ell$-RTVD}), where we look for at most $k$ vertices $S$ such that $D-S$ has no more than $\ell$ transitive arcs. We explore {\sc $\ell$-RTVD} on various well-known graph classes of digraphs such as directed acyclic graphs (DAGs), planar DAGs, $\alpha$-bounded digraphs, tournaments, and their multiple generalizations such as in-tournaments, out-tournaments, local tournaments,  acyclic local tournaments, and obtain the following results. Although the problem admits polynomial-time algorithms in tournaments, $\alpha$-bounded digraphs, and acyclic local tournaments for fixed values of $\ell$, it remains $\textsf{NP}$-Hard even in planar DAGs with maximum degree 6.	In the parameterized realm, for {\sc $\ell$-RTVD} on in-tournaments and out-tournaments, we obtain polynomial kernels parameterized by $k+\ell$ for bounded independence number. But the problem remains fixed-parameter intractable on DAGs when parameterized by $k$.
\end{abstract}


\end{frontmatter}

\section{Introduction}
The act of removing vertices from a directed graph in a manner that ensures specific properties in the resulting graph has a rich history in the field of graph theory. Often, this process serves as a preliminary stage for various graph drawing algorithms, including 3D orthogonal graph drawings \cite{eades2000three}, right-angle-crossing drawings \cite{angelini2011perspectives}, hierarchical and upward drawings \cite{binucci2016computing}. For instance, one may consider Eulerian orientations, where each vertex has an equal in-degree and out-degree, acyclic orientations that yield a directed acyclic graph, or singly connected orientations that guarantee at most one path between any pair of vertices \cite{article}.

Our  paper focuses on investigating a significant property of directed graphs: transitive-freeness. Specifically, we delve into the problem of removing vertices from a graph so that only a limited number of transitive arcs remain in the resulting structure. In a directed graph $D$, an arc $e=(u,v)$ is considered transitive if at least one path exists from $u$ to $v$ in the graph $D-e$ \cite{binucci2023st}. A digraph is said to be {\em transitive-free} if it contains no transitive arcs. In a transitive-free digraph, each arc plays a crucial role in maintaining connectivity that any combination of the remaining arcs cannot achieve. In other words, every arc is essential for preserving connectivity. The optimization version of this problem, referred to as the Transitive-free Vertex Deletion (TVD), aims to find the minimum number of vertices to be deleted from the graph to achieve transitive freeness. We investigate a generalization of this problem known as \lrtfvd (\lvd), where the digraph can have a maximum of $\ell$ transitive arcs.

In many graph layout problems, existing algorithms often rely on an initial orientation computation for the given graph. In such scenarios, reducing the number of transitive arcs in this orientation tends to affect the overall readability of the layout positively. Despite its similarities to the {\sc Singly Connected Vertex Deletion} \cite{article} and {\sc Feedback Vertex Set}  \cite{chen2008fixed} problems, the {\sc Transitive-free Vertex Deletion}(TVD) problem has received relatively little attention. Only recently, Binucci et al. \cite{binucci2023st} investigated a slightly different variant of the problem, which involved orienting edges in an undirected graph to minimize the count of transitive arcs. Their study demonstrated that the problem is \nph for general graphs. In contrast, our paper takes a distinct approach by exploring the vertex deletion version of the problem.

\defparrprob{{\sc $\ell$-Relaxed Transitive-free Vertex Deletion (\lvd)} }{ A digraph $D$, a non-negative integer $k$. }{}{Does there exist a set $S\subseteq V(D)$ such that $|S| \leq k $ and $D - S$ has at most $\ell$ transitive arcs?}

\subsection{Our contribution and methods}

To begin with, in Section \ref{sec-nphard}, we show that \lvd is \nph on directed acyclic graphs (DAGs), even when the underlying undirected graph is planar and has  maximum degree  six. This result is derived through a hardness reduction from {\sc Vertex Cover} on planar graphs with  maximum degree  three \cite{DBLP:journals/jct/Mohar01}.

Subsequently, we derive several polynomial-time algorithms (for constant $\ell$) for various widely recognized subclasses of directed graphs, including tournaments, $\alpha$-bounded digraphs, and acyclic local tournaments. For instance, our algorithm for tournaments (Section \ref{sec:tour}) is constructed by combining Frankl's result on a uniform family of set system \cite{DBLP:journals/jct/FranklF85} with the straightforward observation that any tournament with more than four vertices contains a transitive arc. In the case of $\alpha$-bounded digraphs, we utilize the Gallai-Milgram theorem \cite{gallai1960verallgemeinerung} to establish an upper bound of $\mathcal{O}(\alpha^2)$ on the size of a transitive-free digraph, thereby enabling the development of a polynomial-time algorithm (Section \ref{sec:alphabounded}). Furthermore, employing a dynamic programming approach, we design our third polynomial-time algorithm specifically tailored for acyclic local tournaments (Section \ref{sec:localtour}).

Additionally, we study \lvd in the parameterized framework focusing on various subclasses of digraphs and obtain the following results. In Section \ref{sec-hard}, we show that this problem remains \woh~on directed acyclic graphs (DAGs) when parameterized by the solution size ($k$). To prove this, we utilize a reduction from \textsc{Vertex Multicut} in DAGs, which is already known to be \woh when parameterized by the solution size \cite{multicutindag}.
However, for both in-tournaments and out-tournaments, we  design polynomial (in $k+\ell+\alpha$) size kernels (Section \ref{sec-kernel}), where $\alpha$ is the independence number (equivalently, the graph is $\alpha$-bounded). In designing these kernels, the concept of the {\em cut preserving set} \cite{DBLP:conf/innovations/LochetLM0SZ20} plays a crucial role.

Finally, in Section \ref{Sec-fpt}, we explore the problem for the special case when $\ell=0$ on in-tournaments (and out-tournaments) and show that the problem reduces to the 3-Hitting Set problem, thus admitting an \fpt algorithm and a quadratic vertex kernel.

\section{Preliminaries}

For a positive integer $n$, $[n]$ denotes the natural number set $\{1,2,\ldots,n\}$.  For a digraph, $D$, $V(D)$ and $A(D)$ denote the set of vertices and arcs in $ D $, respectively. A directed arc from $u$ to $v$ is denoted by the ordered pair $(u,v)$.  For a set $S \subseteq V(G)$, $D - S$ denotes the digraph obtained by deleting $S$ from $D$ and $D[S]$ denotes the subgraph of $D$ induced on $S$.  Similarly for a set $A'\subseteq A(D)$, $D-A'$ denotes the digraph obtained by deleting the arcs  $A'$ from $ D $. By $N^{+}_D(v)=\{u|(v,u) \in A(D)\}$ we denote the set of out-neighbors of $v$. Similarly, $N^{-}_D(v)=\{u|(u,v)\in A(D) \}$ denotes the set of in-neighbors of $v$. We may not use $D$ in the notation when the context is clear. A {\em path} $P$ of order $t$ is a sequence of distinct vertices $v_1,v_2,\ldots,v_t$ such that $(v_i,v_{i+1})$ is an arc, for each $i,\ i\in [t-1]$. A {\em path cover} is a set of paths such that every vertex of the digraph $D$ appears on some path. A digraph $D$ is said to be {\em singly connected} if for all ordered pairs of vertices $u,v \in V(D)$, there is at most one path in $D$ from $u$ to $v$.

\begin{definition}[Transitive arc] \label{def:transitivity}
	{\em An arc $e=(u,v)$ is a transitive arc in the digraph $D$ if there exists a directed path $P$ from $u$ to $v$ in $D-e$. A graph $D$  is said to be transitive-free if it has no transitive arcs.} 
\end{definition}

\begin{definition}[Tournament]\label{def:tournament}
	{\em A digraph $D$ is said to be a tournament if for every pair $u,v$ of vertices in $D$, either $(u,v) \in A(D)$ or $(v,u) \in A(D)$ but not both.}
\end{definition}

\begin{definition}[Out-Tournament, In-Tournament, Local-Tournament]\label{def:outinlocal}
	{\em A digraph $D$ is an out-tournament if for every vertex $v$, $D[N^{+}(v)]$ is a tournament. Similarly, a digraph $D$ is an in-tournament if for every vertex $v$, $D[N^{-}(v)]$ is a tournament. $D$ is a local tournament if and only if it is both an out-tournament and an in-tournament.}
\end{definition}

An illustration of examples of definitions \ref{def:tournament} and \ref{def:outinlocal} is depicted in \Cref{fig:enter-label}.

\begin{definition}[$ \alpha $-bounded digraph]\label{def:alpha}
	{\em A digraph $D$ is said to be $ \alpha $-bounded if the size of a maximum independent set
	of the underlying undirected graph of $D$ is at most $\alpha$.} 
\end{definition}

\begin{figure}
    \centering
    \includegraphics[width=1\linewidth]{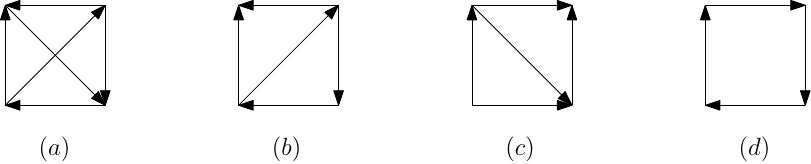}
    \caption{Examples of  tournament (a), in-tournament (b), out-tournament (c), local-tournament (d).}
    \label{fig:enter-label}
\end{figure}


\noindent 
{\bf Parameterized Complexity.} 
A parameterized problem $ \Pi $ is a subset of $ \Gamma^* \times \mathbb{N}$ for some finite alphabet $ \Gamma $. An instance of a parameterized problem consists of $(X, k)$, where $k$ is called the parameter. A parameterized problem $ \Pi \subseteq \Gamma^* \times \mathbb{N}$ is said to be fixed-parameter tractable (or \fpt) if there exists a computable function $f$, a constant $c$, and an algorithm that on input $(X,k) \in  \Gamma^* \times \mathbb{N}$, runs in time $f (k)|(X,k)|^c$, and correctly decides whether or not $(X,k) \in \Pi$. A problem $\Pi$ is \woh~if it does admit  an \fpt algorithm. On the other hand, a kernelization algorithm, for a parameterized problem $ \Pi \subseteq \Gamma^* \times \mathbb{N}$ is an algorithm that, given $(X, k) \in  \Gamma^* \times \mathbb{N}$, outputs in time polynomial in $|X| + k$ a pair $ (X', k')\in  \Gamma^* \times \mathbb{N} $ such that (a) $ (X, k) \in \Pi$ if and only if $(X', k') \in \Pi$ and  (b)  $|X'|,|k| \leq g(k)$, where $ g $ is some computable function depending only on $ k $.  The output of kernelization $  (X', k') $ is referred to as the kernel and the function $ g $ is referred to as the size of the kernel. If $ g(k) \in k^{\mathcal{O}(1)}  $ , then we say that $ \Pi $ admits a polynomial kernel.  We refer to the monographs \cite{DBLP:series/mcs/DowneyF99,DBLP:series/txtcs/FlumG06,DBLP:books/ox/Niedermeier06} for a detailed study of the area of kernelization.

\section{NP-completeness on bounded degree planar DAGs}\label{sec-nphard}

In this section, we show that \lvd remains {\sf NP}-complete     even on planar DAGs with maximum degree $6$.

\begin{theorem}\label{theo-planardag}
\lvd is {\sf NP}-complete on planar directed acyclic graphs with maximum degree $6$.
\end{theorem}

 We give a reduction from the well-known \nph problem of {\sc Vertex Cover} on 2-connected cubic planar graphs \cite{DBLP:journals/jct/Mohar01}  to \lvd on planar DAGs of maximum degree 6. Consider an instance $(G,k)$ of {\sc Vertex Cover} on  2-connected cubic planar graph. We construct an equivalent instance $(H, k)$ of  \lvd on planar DAGs with maximum degree $6$ in the following manner. Let $V(G)= \{v_1, v_2, \ldots , v_n\}$.

 \begin{itemize}
     \item Initialize $V (H) = V (G)$. 
     \item For each edge  $v_i v_j \in E(G)$ where $i < j$, we add a vertex $x_{ij} $ in $V(H)$ and add the arcs $(v_i, x_{ij})$, $ (x_{ij},v_j)$ and $(v_i, v_j)$ in $A(H)$. Let $A_G(H)= \{(u,v) | (u,v) \in A(H)~\text{and}~(u,v) \in E(G)\}$.
 \end{itemize}

 \begin{figure}[ht!]
     \centering
     \includegraphics[width=.6\linewidth]{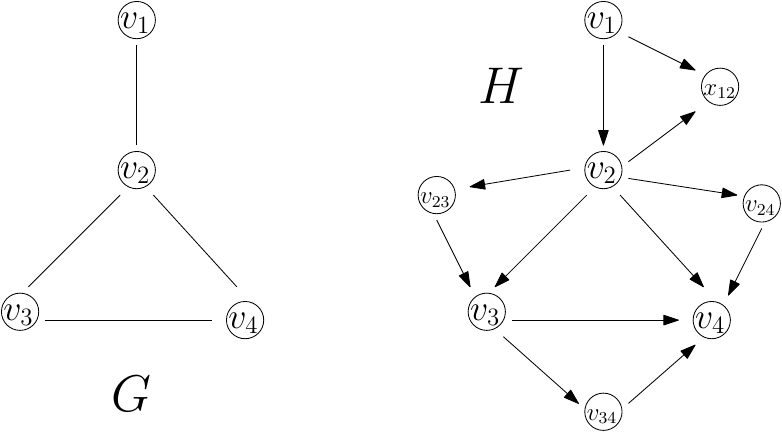}
     \caption{Example of construction for \Cref{theo-planardag}.}
     \label{fig:hard}
 \end{figure}
 
 This completes the construction of the graph $H$. For an illustration of the construction, see \Cref{fig:hard}. Observe that the number of vertices in $H$ is polynomial in $n$. In addition, $H$ is a planar DAG with a maximum degree of 6, and construction can be carried out in polynomial time. It is easy to verify that \lvd is in {\sf NP} (by checking whether there are at most $\ell$ transitive arcs).  Next, we claim the following.

\begin{clam}\label{clm:bounded}
$G$ has a vertex cover of size at most $k$ if and only if $H$ has vertex subset $X \subseteq V(H)$ of size at most $k$ such that $H-X$ has no transitive arcs.
\end{clam}

\begin{proof}
	In the forward direction, suppose $G$ has a vertex cover $S \subseteq V(G)$ of size at most $k$. Then, $G - S$ is an independent set. Now each arc of $H-S$ is either of the form $(v_i, x_{ij})$ or $(x_{ij}, v_j)$ for some $i,j \in [n]$. But none of these arcs are transitive in $H$ and consequently also in $H-S$. Hence $S$ is a vertex subset of $H$ such that $G-S$ has no transitive arcs.

	In the backward direction, suppose $X\subseteq V(H)$ is of size at most $k$ such that $H-X$ has no transitive arcs. 
	First, we try to construct a set $X'\subseteq V(G)$ of size at most $k$ such that $H-X'$ has no transitive arcs in addition to the fact that $X'$ only contains vertices from $V(G)$. Suppose $X$ has a  vertex $v$ that is not from $V(G)$. Let $v=x_{ij}\in (X\cap V(H))\setminus V(G)$. We replace the vertex $v$ with $v_i$ in $X$ and claim that it still remains a solution. If our claim holds true, then we can apply this procedure exhaustively to construct the desired $X'$. Otherwise let $e=(a,b)$ be a transitive arc in $H-((X\setminus\{x_{ij}\})\cup \{v_i\})$. Notice that $e=(a,b)\in A_G(H)$ as only such arcs can be transitive. But then there is a path from $a$ to $b$ in $H-((X\setminus\{x_{ij}\})\cup \{v_i\})$, while there was no path from to $a$ to $b$ in $H-X$ which is not possible since any path through $x_{ij}$ also passes through the vertex $v_i$. Hence a contradiction. So $H-X'$ has no transitive arcs. Since $X'\subseteq V(G)$, it must be a vertex cover. If an edge $ab \in E(G)$ is not covered by $X'$, then $(a,b)$ becomes a transitive arc in $H-X'$, a contradiction. Hence $X'$ is a vertex cover of $G$.\end{proof}

\noindent This completes the proof of Theorem \ref{theo-planardag}, since for the claim to hold for $\ell$ at least 1, it suffices to add in $H$, for some $v_i$ in $H$, vertices $v'_1, \ldots ,v'_{\ell}, x'_1, \ldots, x'_{\ell},v^*_1, \ldots, v^*_{\ell}$, such that $(v'_j,v^*_j), (v'_j,x'_j), (x'_j,v^*_j)$ and $(v^*_j,v_i)$ for all $j$ in $[\ell]$.\qed

\section{Polynomial-time Algorithms}\label{sec-poly}

We study the optimization version of \lvd, i.e., delete minimum number of vertices such that  the resulting graph has at most $\ell$ transitive arcs, which we call \mlrtvd (shortly, \mvd) on some well-known restrictions of digraphs such as tournaments, $\alpha$-bounded digraphs, and acyclic local tournaments. We obtain polynomial-time algorithms on all these graph classes.

\subsection{\mvd on Tournaments}\label{sec:tour}

\begin{theorem} \label{theo:tournatheorem}
	\mvd on tournaments with $n$ vertices admits an algorithm running in $\mathcal{O} (n^{4\sqrt{\ell}})$ time.
\end{theorem}

We start with the following simple observation on tournaments.

\begin{lemma} \label{lem:fourvertextourna}
	Any tournament with at least  $4$ vertices has a transitive arc.
\end{lemma}

\begin{proof}
	We prove this lemma using contradiction. Consider a tournament $T$ on $4$ vertices, say $ v_1, v_2, v_3$ and $v_4 $ that does not have any transitive arc. Let  $(v_1, v_2) \in A(T)$. Neither $ T[\{v_1, v_2, v_3\}] $ nor $ T[\{v_1, v_2, v_4\}] $ are acyclic triangles, otherwise $T$ has a transitive arc. So $T$ has the arcs $ \{(v_1, v_2), (v_2, v_3), (v_3, v_1), (v_2,v_4), (v_4, v_1) \} $ in $ A(T) $. Now if $ (v_4, v_3)  \in A(T)$, then the arc $ (v_2, v_3) $ becomes transitive and  if  $ (v_3, v_4)  \in A(T)$, then the arc $ (v_2, v_4) $ becomes transitive, a contradiction. Hence any tournament with at least 4 vertices  contains a transitive arc.
\end{proof}

Using Lemma \ref{lem:fourvertextourna} along with the hereditary property of tournaments, we get the following corollary.

\begin{corollary}\label{cor:tour}
	Any tournament with at least  $(4\ell+4)$ vertices contains at least $(\ell+1)$ transitive arcs.	
\end{corollary}

\noindent 	We can improve upon this bound with the following combinatorial result.

\begin{proposition}[Frankl \cite{DBLP:journals/jct/FranklF85}]\label{propo:frank}
	Given a universe $U$ of $n$ elements, for every $k \geq m \geq 1$ and $n \geq k^2$, there exists a $k$-uniform family $\mathcal{F}$ (where each set of $\mathcal{F}$ has size exactly $k$) of size at least  $ (n/k)^m$  such that $|A \cap B| \leq m-1$ for any two distinct sets $A, B \in \mathcal{F}$.
\end{proposition}

Let $T$ be a tournament with $4\sqrt{\ell}$ vertices. By setting $V(T)=U$, $k=4$, $m=2$ and applying Proposition \ref{propo:frank}, we get a family $\mathcal{F}$ of size at least $\ell$. Notice that any set in $\mathcal{F}$ induces a  tournament on 4 vertices and hence has a transitive arc contained inside (by Lemma \ref{lem:fourvertextourna}). Moreover, no two tournaments corresponding to two distinct sets in $\mathcal{F}$ can have the same transitive arc since they have at most one vertex in their intersection. Hence $T$ has at least $\ell$ many transitive arcs.  

\begin{corollary}
	Any tournament with at least  $4\sqrt{\ell}$ vertices has at least $\ell$ transitive arcs.
\end{corollary}

\noindent Now given a tournament $T$, solving \mvd on $ T $  reduces to finding a maximum size set $ V' \subseteq V(T) $ such that $ D[V'] $ has at most $\ell$ transitive arcs. Since $ |V'| < 4\sqrt{\ell} $, one can find such a set in time $\mathcal{O} (n^{4\sqrt{\ell}})$. 

This completes the proof of Theorem \ref{theo:tournatheorem}.

\subsection{\mvd on $\alpha$-bounded digraphs}\label{sec:alphabounded}

We start with an algorithm obtained by the application of Ramsey’s theorem \cite{cameron1994combinatorics}: for positive integers $r$ and $s$,  every graph on $R(r,s) \leq  { r+s-2 \choose r-1}$  vertices contains either a clique of size $r$ or an independent set of size $s$.  So for any digraph $D$, if the number of vertices is more than ${ 4+(\alpha+1)-2 \choose 4-1} = \mathcal{O}(\alpha^3)$, then either it has a clique of size 4 or an independent set of size $ (\alpha+1) $. In Lemma  \ref{lem:fourvertextourna}, we proved that any tournament with at least  $4$ vertices is not transitive-free. Since by definition any $\alpha$-bounded digraph has no independent set of size $ (\alpha+1) $, if $ |V(D)| \geq  \mathcal{O}(\alpha^3)$, then $ D $ must have a clique of size $4$ and consequently has a transitive arc. In this section, we improve this bound to $ \mathcal{O}(\alpha^2) $ and obtain the following result.
\begin{theorem} \label{theo:bounded}
	\mvd on $\alpha$-bounded digraphs with $n$ vertices admits an algorithm running in  $\mathcal{O}(n^{\alpha (2\alpha+3) (\ell+1) +\ell})$ time.
\end{theorem}
We start by defining \emph{forward} and \emph{backward} arcs along with presenting a couple of simple observations.
\begin{definition}[Forward and backward arc]
	Let $P=v_1,\ldots,v_{\ell}$ be a path in the digraph $D$. An arc $e=(v_i,v_j)$, where $ v_i, v_j \in V(P)$ is said to be a forward arc (resp. backward arc) in the digraph $D$ with respect to the path $P$ if  $ j > (i+1) $ (resp. $ i > (j+1) $).
\end{definition}

\begin{observation}\label{obs:1}
	Let $D$ be a transitive-free digraph and $P$ be a path in $D$. Then  every vertex of $P$ cannot be neither the endpoint of any forward arc nor two backward arcs in $D$ with respect to the path  $P$.
\end{observation}

\begin{proof}
	The proof follows naturally from the fact that if a vertex is the endpoint of two  backward arcs or one forward arc then a transitive arc exists.
\end{proof}

\begin{observation}\label{obs:2}
	Let $D$ be a transitive-free digraph and $P$ be a path in $D$. Then  any vertex of $V(D)-V(P)$  has at most two neighbors (both in-neighbor and out-neighbor) in $P$.
\end{observation}

\begin{proof}
	If there is a vertex $ x \in  V(D)-V(P)$ with more than two neighbors in $ V(P) $, then $ x $ has either more than one in-neighbor or more than one out-neighbor. In either case, we get a transitive arc, a contradiction.
\end{proof}

Combining Observations \ref{obs:1} and \ref{obs:2}, we have the following lemma.

\begin{lemma}\label{lem:neighbour}
	Let $D$ be a transitive-free digraph and $P$ be a path in $D$. Then any vertex in $P$ has at most $3$ neighbours in $V(P)$ and any vertex in $V(D)-V(P)$ has at most 2 neighbours in $V(P)$.
\end{lemma}

Along with  Lemma \ref{lem:neighbour}, we use the following theorem to prove the promised improved bound on the size of $D$.

\begin{proposition}[Gallai-Milgram \cite{gallai1960verallgemeinerung}]\label{theo-gallai}
	Every digraph $D$ has a path cover $\mathcal{P}$ and an independent set $\{ v_P | P \in \mathcal{P} \}$ of vertices such that $ v_P \in P $ for each $ P \in \mathcal{P} $. 
\end{proposition}

\begin{lemma}\label{lem:alpha}
	Let $ D $ be an $\alpha$-bounded digraph. If $ |V(D)| > \alpha(2\alpha+3)$, then $ D $ is not transitive-free. 
	
\end{lemma}

\begin{proof}
	Let $\mathcal{P}$ be a path cover and $ I $ be a maximal independent set  on $ D $ obtained by the algorithm of  \Cref{theo-gallai}. Since $ D $ is $ \alpha $-bounded, $ |\mathcal{P}|\leq \alpha $. Now if $ |V(D)| > \alpha(2\alpha+3)$, then by the pigeonhole principle, there exists a path $ P \in \mathcal{P} $ such that $ |V(P)| \geq (2\alpha+4)$.  Let $ v \in I \cap V(P) $ and $ S = V(P) \setminus N_P(v) $. By Lemma  \ref{lem:neighbour}, we have $ |S| \geq (2\alpha -1) $. As $ I $ is maximal, every vertex in $ S $ has a neighbor in  $ I \setminus \{v\} $. As $| I \setminus \{v\}|\leq\alpha-1 $, by the pigeonhole principle, there is a vertex in $I \setminus \{v\}  \subseteq (V \setminus V(P))$, which is adjacent to at least three vertices in $ S \subseteq V(P) $, which contradicts Lemma \ref{lem:neighbour}. This completes the proof.	
\end{proof}

Using Lemma \ref{lem:alpha} and the fact that $\alpha$-bounded digraphs are hereditary, we get the following corollary.

\begin{corollary}
	Any $\alpha$-bounded digraph with more than $ \alpha(2\alpha+3) (\ell +1) + \ell $ vertices contains at least $(\ell+1)$ transitive arcs.	
\end{corollary}

\noindent Now given a $\alpha$-bounded digraph $D$, solving \mvd on $ D $ reduces to finding a maximum size set $ V' \subseteq V(D) $ such that $ D[V'] $ has at most $\ell$ transitive arcs. Since $ |V'| \leq (\alpha(2\alpha+3)+1) \ell + \alpha(2\alpha+3) = \alpha (2\alpha+3) (\ell+1) +\ell $, we can find such a set in $\mathcal{O}(n^{\alpha (2\alpha+3) (\ell+1) +\ell})$. This completes the proof of Theorem \ref{theo:bounded}.\qed

\subsection{\mvd on Acyclic Local Tournaments}\label{sec:localtour}

In this section, we prove the  following result.

\begin{theorem}\label{theo-acyclic}
	\mvd on  acyclic local tournaments with $ n $ vertices admits an algorithm running in $ n^{\mathcal{O}(\ell)} $ time.
\end{theorem}

Below we present several facts and certain lemmas along with Theorem \ref{algo} that play a vital role in the development of our algorithm, as outlined in Theorem \ref{theo-acyclic}.

\begin{fact}[\cite{bang2008digraphs}]\label{fact:one}
 Let $D$ be an acyclic digraph. Then $D$ has a topological ordering. That is, there exists an ordering $\sigma= (v_1, \ldots, v_n)$ of the
vertices of $D$ such that for every arc $(v_i, v_j) \in A(D)$, we have $i < j$, i.e., $v_i$ appears before $v_j$ in the ordering $\sigma$. Moreover, there exists
an algorithm that, given an acyclic digraph $D$ with $n$ vertices and $m$ arcs as input, runs in time $\mathcal{O}(n + m)$, and finds a topological
ordering of $D$.
\end{fact}

\begin{fact}[\cite{article}] \label{fact:uniqtoporderifham}
	An acyclic digraph $D$ has a unique topological ordering if and only if $D$ has a Hamiltonian path.
\end{fact}

\begin{fact}[\cite{article,bang1990locally}] \label{fact:loctournaimpham}
	Let $D$ be a connected local tournament. Then $D$ contains a Hamiltonian path.
\end{fact}

\begin{lemma}[\cite{article}]\label{lem:topo_ordering_ALT}
	Let $D$ be a connected acyclic local tournament and $P = v_1 v_2 \ldots v_n$ be a Hamiltonian path in $D$. Then  $\sigma= (v_1, \ldots, v_n)$ is 	the unique topological ordering of $ D $.
\end{lemma}

Let $D$ be a connected (underlying undirected graph is connected) acyclic local tournament and   $\sigma= (v_1, \ldots, v_n)$ be its  unique topological ordering. For a vertex $ v  $ in $ D $, $ \ell(v) $ denotes the last vertex $ u $ in the ordering $ \sigma $ such that $( v, u ) \in  A ( D )$. For each $i \in  [ n ]$ , we define an ordered set $S_i = \{ v_i , v_{i+1}, \ldots , \ell(v_i)\} $.

\begin{lemma}[\cite{article}]\label{lem:S_i_tour}
	Let $D$ be a connected acyclic local tournament and  $\sigma= (v_1, \ldots, v_n)$ be its unique topological ordering. Then, for all $i \in [n]$,
	the graph $D [ S_i ]$ is an acyclic tournament. Moreover, $ S_i= N_D^+(v_i) \cup \{v_i\} $.
\end{lemma}


\begin{lemma}\label{lem:tfsc}
	Any  local tournament $D$ is  transitive-free if and only if $D$ is singly connected.	
\end{lemma}

\begin{proof}
	In the forward direction, suppose that $D$ has no transitive arcs and is not singly connected. Then, there exists a pair of vertices $u$ and $v$ such that there are at least two vertex-disjoint paths of length at least two  from $ u $ to $ v $ in $ D $. Let $ P_1 $ and $ P_2 $ be two such paths. But, then $u$ has two out neighbours, say $ y_1$ in $P_1$ and $y_2$ in $P_2$. As $ D $ is a local tournament, either $ (y_1, y_2) \in A(D) $ or $ (y_2, y_1) \in A(D)$. In both these cases, we get either $(u,y_1)$ or $(u,y_2)$ as transitive arcs, which is a contradiction. The reverse direction is immediete. This completes the proof. 
\end{proof}

Below we state a crucial result, that help  to design  our algorithm.

\begin{theorem}[\cite{article}]\label{algo}
	Let $D$ be a connected acyclic local tournament on $ n $ vertices and  let $\sigma= (v_1, \ldots, v_n)$ be the unique topological ordering of $ D $. In $\mathcal{O}(n^3)$ time, we can find a  minimum-size set $S$ such that $D-S$ is transitive-free. Additionally, we have $v_1\notin S$. 
\end{theorem}

We refer to the algorithm of Theorem \ref{algo} as ``Algo$\_$Singly$ \_ $Connected$(D)$''  and use it as a subroutine in Algorithm \ref{algo_1}. Next, we make a
two crucial observations below  for  our algorithm.



\begin{sloppypar}

\begin{observation}\label{obs2}
		Let $X$ be an optimal solution to \mvd on $ D $. For a pair of integers $ i ,j \in [n] $, if $ (v_i, v_j) $ is an transitive arc in $ D-X $ and no arc of the form $ (v_p, v_q) $ where $ i\leq p<q \leq  j $ is transitive in  $ D-X $, then $|\{v_{i+1}, v_{i+2}, \ldots, v_{j-1}  \} \setminus X| \leq 1$.	That means  all but at most one vertex of the set $\{v_{i+1}, v_{i+2}, \ldots, v_{j-1}  \}$	belongs to $X$.
	\end{observation}
	
	\begin{proof}
		Let there be two vertices $v_{p}$ and $ v_{q} $ in $ \{v_{i+1}, v_{i+2}, \ldots, v_{j-1}  \} \setminus X $  with $p < q$. Then applying Lemma \ref{lem:S_i_tour}, $D[v_i, v_{p}, v_{q}]$ forms a tournament in $D-X$ and that indeed produce a transitive arc $ (v_i, v_{q}) $ in $D-X$, a contradiction. 
	\end{proof}

    \begin{observation}\label{obs3}
		Let $X$ be an optimal solution to \mvd on $ D $ and $v_p \notin X$ be a vertex that is not incident to any transitive arc in $D-X$. Then  for any pair of vertices   $ v_i, v_j \notin X $  with $i <j$ either of the following two conditions  can not occur.\\
		(1) $ p <i $ and $ (v_p,v_j) \in A(D)$, \\
		(2) $ j <p $ and $ (v_i,v_p) \in A(D)$.
		\end{observation}

	\begin{proof}
		Let $ p <i $ and $ (v_p,v_j) \in A(D)$, then by  applying Lemma \ref{lem:S_i_tour} the arc  $ (v_p,v_j) \in A(D)$ become a transitive arc in $D-X$, a contradiction. In a similar manner, if $ j <p $ and $ (v_i,v_p) \in A(D)$, then by  applying Lemma \ref{lem:S_i_tour} the arc  $ (v_i,v_p) \in A(D)$ become a transitive arc in $D-X$, a contradiction. 
	\end{proof}

    \IncMargin{1em}
\begin{algorithm}[ht!]
	\SetKwData{Left}{left}\SetKwData{This}{this}\SetKwData{Up}{up}
	\SetKwFunction{Union}{Union}\SetKwFunction{FindCompress}{FindCompress}
	\SetKwInOut{Input}{Input}\SetKwInOut{Output}{Output}
	\Input{A connected acyclic local tournament $D$ with the topological ordering $\sigma= (v_1, \ldots, v_n)$ of $V(D)$.}
	\Output{A solution $S$ to \mzvd such that  $v_1, v_n \notin S$.}
	\BlankLine
	
	\If{$ D $ has no transitive arcs,}{\Return $ S=\emptyset$\;}
	\Else
	{
		$S^*=$ Algo$\_$Singly$ \_ $Connected$(D)$.
		
		\If{$ v_n \notin S^{*} $, }{\Return $ S=S^{*}$\;}
		
		\Else
		{
			Let $ q $ be the largest integer such that $ v_q \notin S^{*}$ and $ (v_q,v_n) \in A(D)$. \Return $ S=(S^{*} \setminus \{v_n\}) \cup \{v_q\}$\;
			
		}
		
	}
	
	\caption{ Algo$\_$Transitive$(D)$}\label{algo_1}
\end{algorithm}\DecMargin{1em}

Next we show that there is always  a solution to  \mzvd 
 avoiding both the vertices $v_1$ and $v_n$.
	
	\begin{lemma}\label{lem-first}
		Let $D$ be a connected acyclic local tournament and $\sigma= (v_1, \ldots, v_n)$ be its unique topological ordering. Then there is an optimal
		solution to \mzvd  on $D$ such that the solution contain no vertex from  $\{v_1,v_n\}$. Moreover we can find such a solution in  $\mathcal{O}(n^3)$ time.
	\end{lemma}
	\begin{proof}
		Let $ S $ be an optimal solution to \mzvd on $ D $. Below we consider the two cases depending on whether $v_1 \in S$ or $v_n \in S$.
		
		\noindent \textbf{Case 1.} $ v_1 \in S $.\\ Consider the digraph $ D-(S\setminus \{v_1\}) $. Since $ S $ is an optimal solution, there must exist a transitive arc in $ D-(S\setminus \{v_1\}) $. Moreover $ v_1 $ is one of the endpoints of this transitive arc. But $ v_1 $ has at most two out neighbours in $ D-(S\setminus \{v_1\}) $. Suppose $ p $ is the smallest integer such that $ (v_1, v_p) \in A(D) $ and $ v_p \in D-S $. Clearly,	we do not have any transitive arcs in $ (S\setminus \{v_1\}) \cup \{v_p\}$. But, then $ (S\setminus \{v_1\}) \cup \{v_p\}$ is also an optimal solution.
		
		\noindent \textbf{Case 2.} $ v_n \in S $. \\ Consider the digraph $ D-(S\setminus \{v_n\}) $. Since $ S $ is the optimal solution then there must exist a transitive arc in $ D-(S\setminus \{v_n\}) $' Moreover  $ v_n $ is one of the endpoints of this transitive arc. But $ v_n $ has at most two in neighbours in $ D-(S\setminus \{v_n\}) $.  Suppose   $ q $ is the largest integer such that $ (v_q, v_n) \in A(D) $ and $ v_q \in D-S $. 	Clearly, we do not have any transitive arcs in $ (S\setminus \{v_n\}) \cup \{v_q\}$. But, then  $ (S\setminus \{v_n\}) \cup \{v_q\}$ is also an optimal solution.	
        
       According to \Cref{lem:tfsc}, a local tournament $D$ is transitive-free if, and only if, it is singly connected. Therefore, we can utilize the algorithm from  \Cref{algo}. It's certain that this algorithm will yield a result that excludes $v_1$, although it might include $v_n$. If this occur, we can employ the reasoning from Case 2 of the current proof to obtain an alternative solution that excludes both $v_1$ and $v_n$ (refer to Algorithm \ref{algo_1}). \end{proof}
	
\end{sloppypar}

\noindent \textbf{Overview of our algorithm for \Cref{theo-acyclic}.} Consider $O$ as an hypothetical solution we are looking for. Let $F$ represent the collection of all transitive arcs in $D-O$. It is evident that $|F| \leq \ell$. We will guess the arc set $F$. Notice that for each arc $(v_i, v_j) \in F$ all but at most one vertex of $\{v_{i+1}, v_{i+2}, \ldots, v_{j-1}\}$ must belongs to $O$. So in total at most $\ell$ vertices are there  from the set  $\{v_{i+1}, v_{i+2}, \ldots, v_{j-1}; (v_i, v_j) \in F\}$ that is not in solution. We can guess this set of vertices, denote it as $W_1$. Now we have to devise a solution excluding $W_1$ and the vertices that are incident to  $F$. From all potential selections of $F$ (which total $n^{\mathcal{O}(\ell)}$), we choose the solution with the smallest size. Below we present our algorithm more formally. 

There are two-phases of the algorithm, in the first phase we will reduce our problem to another problem referred as \mvdext (shortly, \mvde). In the second phase, we solve \mvde. 

\noindent \textbf{Phase 1.} Phase one consists of two steps; guess an arc set and then create an instance of \mvde.

\begin{description}
    \item[\textit{Step 1.1.}] Guess an arc subset $ F \subseteq A(D) $ of cardinality at most $ \ell $ and an vertex subset $X \subseteq V(D)$ of size at most $\ell$ such that every transitive arc in $D[W]$ belongs to $F$, where  $W=X \cup V_F$, $V_F$ denotes the vertices in $D$ which are incident  to $F$. Clearly $|W| \leq 3\ell$.

        \item[\textit{Step 1.2.}] Find a minimum-sized vertex set $S \subseteq V (D-W)$ such that  $D - S$ has no transitive arc in $ A(D) \setminus F $.
\end{description}

	
	
	
	

\noindent \textbf{Phase 2.} Phase two consists of procedure to handle the Step 1.2. of Phase 1. Towards that we define an auxiliary problem, referred as  \mvde, which formally defined as follows:

\medskip 

\defparrrprob{\mvde}{ A digraph $D$, $W \subseteq V(D)$, $|W| \leq 3\ell$, $F \subseteq A(D[W]) $, $ |F|\leq \ell$.}{A minimum-sized vertex subset  $S \subseteq V (D-W)$ such that  $D - S$ has no transitive arc in $ A(D) \smallsetminus F $.}

\medskip 

We solve  \mvde to the instance $(D, W, F)$ in  following four steps:

\begin{description}

    \item[\textit{Step 2.1. (Finding  necessary vertices)}:]  We construct a vertex subset $O_{N}$ as follows. 
    \begin{itemize}
        \item[-] For each arc  $ e_{ij}= (v_i, v_j) $ in $F$, we add all the vertices $ \{ v_t ~:~ i <t <j\} \smallsetminus  W$ to $ O_{N} $. 

        \item[-] For every pair of vertices    $ v_p, v_q \in W $  with $p <q$ if there is a vertex $v_t$ such that either (i) $ p <i $ and $ (v_p,v_j) \in A(D)$, or (ii) $ j <p $ and $ (v_i,v_p) \in A(D)$ holds then add $v_t$ to $O_N$.
    \end{itemize}

    \item[\textit{Step 2.2. (Reducing to $\mathcal{O}(\ell)$ many  instances of \mzvd)}:] We create a set of at most $ 3\ell+1 $ instances of  \mzvd in the following manner.
    
    \begin{itemize}
        \item[-] Consider the subgraph $D-O_{N}$

        \item[-]  find a set $\mathcal{P}$ of all possible pairs of integers $ i , j $ in $[n]$ with  $ i<j$ such that  $ \{v_{i+1}, v_{i+2}, \ldots, v_{j-1}\} \cap W = \emptyset$ and one of the following conditions hold 
	\\ (i) $v_i, v_j \in W$, 
	\\(ii) $ i=1, v_i \notin  W, v_j \in W$, 
	\\ (iii) $ j=n, v_i \in W, v_j \notin W $.

    \item[-]  For each element $(i, j) $ in $\mathcal{P}$ we create an instance $I_{i,j}$ for  \mzvd as follows. $I_{i,j} = \bigl\{v_t~:~i \leq t \leq j\bigr\} \smallsetminus O_{N}$. 
    \end{itemize}
    
  Given that $|W| \leq 3\ell$, the maximum size of $\mathcal{P}$ is $3\ell +1$. Consequently, we can generate up to $3\ell +1$ instances $I_{i,j}$ for \mzvd. For notational convenience, we denote these new instances as $ I_1, I_2, \ldots, I_{|\mathcal{P}|} $.

    \item[\textit{Step 2.3. (Solving \mzvd for each $ I_i $)}:] Now we solve \mzvd for each $I_i$, where $i \in [|\mathcal{P}|]$. For that we call Algo$\_$Transitive$(D[I_i])$. A formal description of our algorithm called as  Algo$\_$Transitive$(D)$ for an instance $ D $ is given in Algorithm \ref{algo_1} with its correctness in Lemma \ref{lem-first}.  Let $ O_i=$ Algo$\_$Transitive$(D[I_i])$ for each $ i $.

    \item[\textit{Step 2.4. (Return)}:] $ O_{N} \cup O_1 \cup \ldots \cup O_{|\mathcal{P}|}$.
\end{description}

\medskip

\noindent \textbf{Correctness.}
Let $O$ be a hypothetical solution. And $F$ be the set of transitive arcs in $D-O$. As our guess is exhaustive, we take the set $F$ as a guess. Now, since $F$ is transitive in $D-O$ no vertex incident to $F$ is part of $X$. Moreover, for each arc $(v_i, v_j) \in F$ with $i<j$, we must have exactly one vertex in the set $\{v_{i+1}, v_{i+2}, \ldots, v_{j-1}  \}$ that does not belong to $X$ (by \Cref{obs2}).  Let $W_1$ denote the set of those vertices in the set $\{v_{i+1}, v_{i+2}, \ldots, v_{j-1} : (v_i, v_j)  \}$ that are not part of $X$. Clearly $|W_1| \leq k$. As our guess is exhaustive, in Step 1.1 of Phase 1, we consider $W_1$ as $X$. The correctness of Step 2.1 follows from \Cref{obs2} and \Cref{obs3} as these two observations determine the vertices that must be part of solution. Now, it remains to show that the combination solution to the instances $I_i$ is valid. Towards, contradiction let $ O'=O_{N} \cup O_1 \cup \ldots \cup O_{|\mathcal{P}|}$ is not a solution. That means that there is a transitive arc in $D-O'$ that does not belong to $F$. Take such a transitive arc $(v_i, v_j)$ with the assumption that $j$ is minimum among all the indexes such that $(v_i, v_j)$ is transitive in $D-O'$. As it is transitive, there must be some vertex $v_t$ such that $i<t<j$ and $(v_i, v_t), (v_t,v_j) \in A(D)$. Now the vertex $v_t$ must belong to $W$, otherwise all three vertices of $\{v_i, v_t, v_j\}$ must be part of the same instance $I_i$ for some $i \in [|\mathclap{P}|$, a contradiction. 

We know that $(\{v_{i+1}, v_{i+2}, \ldots, v_{j-1}   \} \setminus O') \cap W \neq \phi$, but more importantly we show some stronger statement. We show that $\{v_{i+1}, v_{i+2}, \ldots, v_{j-1}   \} \setminus O' = \{v_t\}$. In contrast, if there are another vertex $v_{t'} \in \{v_{i+1}, v_{i+2}, \ldots, v_{j-1}   \} \setminus O'$ then $(v_i, v_{t'})$ is a transitive arc in $D-O'$, a contradiction to the assumption that $j$ is the least integer such that $(v_i,v_j)$ is transitive in $D-O'$.  

Now, as $v_t \in W$, either $v_t$ is adjacent to some transitive arc in $F$ or there is a pair of vertices $v_p$ and $v_q$ such that $p<t<q$ with $(v_p, v_t), (v_t,v_q), (v_p,v_q) \in A(D)$, i.e, $(v_p, v_q)$ is transitive in $D$, moreover $(v_p, v_q) \in F$. We argue for both the cases. First, consider that $v_t$ is part of some transitive arc where the other end point is $(v_t,v_{t'})$. Clearly, $v_{t'} \in W$ and $t' \notin \{i+1, \ldots, j-1\}$. If $(v_t, v_{t'}) \in A(D) \cap F$, then it implies $v_j \in W$ (for the right guess). That makes \Cref{obs3} applicable and adds $v_i$ to $O_N$. If $(v_{t'}, v_t) \in A(D) \cap F$, then it implies $v_i \in W$ (for the right guess). That makes \Cref{obs3} applicable and adds $v_j$ to $O_N$.  Now we consider the case where $v_t$ is not part of any transitive arc in $F$, but there is a pair of vertices $v_p$ and $v_q$ such that $p<t<q$ with $(v_p, v_t), (v_t,v_q), (v_p,v_q) \in A(D)$, i.e., $(v_p, v_q)$ is transitive in $D$, moreover $(v_p, v_q) \in F$.
Clearly $p<i<j<q$. As $(v_p,v_q) \in F$ and $v_t \in W$ then both $v_i$ and $v_j$ must be added to solution $O_N$. Hence $O'$ is a feasible solution. Now the optimality is correct as we are considering all the instances $I_i$ separately and size of the optimal solution for one does not affect to the other. 

\medskip

\noindent \textbf{Running Time.}
According to \Cref{fact:one}, the unique topological ordering of $D$ can be determined in $\mathcal{O}(n + m)$ time. The algorithm comprises two stages. In the first stage, we guess an arc set of size no greater than $\ell$ and generate an instance $(D, W, F)$ of \mvde. As the number of guesses is bounded by $m^{\ell}$, we form up to $m^{\ell}$ instances of \mvde in total. During the second stage, we solve each of these $m^{\ell}$ instances. Let's examine stage two. In Step 2.1, finding the required vertex tales $\mathcal{O}(\ell n)$ time. In Step 2.2, we form a maximum of $3\ell + 1$ instances of \mzvd. Since \mzvd can be solved in $\mathcal{O}(n^3)$ time (as per \Cref{lem-first}), the overall running time is $n^{\mathcal{O}(\ell)}$. \\  This concludes the proof of \Cref{theo-acyclic}. \qed

\section{W[1]-completeness on DAGs} \label{sec-hard}

This section shows that \zvd is {\sf W[1]}-complete when parameterized by the solution size ($k$) on DAGs.

\begin{theorem} \label{theo:W1-HARD-DAG}
	\zvd is {\sf W[1]}-complete on DAGs when parameterized by $k$.
\end{theorem}

	It is easy to verify that \zvd is in {\sf W[1]} (by checking whether there are no path from $u$ to $v$ for each arc $(u,v) \in A(G)$).  Now, we give a polynomial-time parameter preserving reduction from {\sc Vertex Multicut} in DAGs (which is known to be \woh \cite{multicutindag}) parameterized by the solution size. Let $(G,\mathcal{T},k)$ be an instance of {\sc Vertex Multicut} where the graph $G$ is a DAG,  $\mathcal{T}=\{(s_i,t_i)|1\leq i \leq r\}$ is a set of terminal pairs where $(s_i,t_i) \in V(G)$ and  $k$ is a non-negative integer. The goal is to find a vertex subset $S \subseteq V(G)\setminus V(\mathcal{T})$ such that $ |S | \leq k $ and $G-S$ has no paths from $s_i$ to $t_i$ for each $ i \in [r]$. We can assume that  for each $ i \in [r]$, $(s_i,t_i) \notin A(G)$ and there is a path from $ s_i $ to $ t_i $ in $ G $.
	
	First, we describe the construction of an instance $(D,k)$
	of \zvd. Let $V (G) = \{v_1 , \ldots , v_n \}$. To construct the graph $D$, we apply the following procedure.
	
	\begin{enumerate}
		\item Initialize $V (D) = V (G)$, i.e. we add the vertices $v_i$'s in $V(D)$ for each $i \in [n]$. We denote this set of  vertices as $V_G$.
		
		\item For each pair of  terminals $(s_i,t_i)$ where $i \in [r]$, we add the arc $(s_i,t_i)$ in $A(D)$. 
		
		\item For each arc $e=(u,v) \in A(G)$ which is not an arc between a pair of terminals $(s_i,t_i)$, we add a new vertex $e_{uv}$ to $V(D)$ and add two arcs 		$(u,e_{uv})$ and $(e_{uv},v)$ to $A(D)$. Let $ V_E $ be the set of vertices that we add in this step. Notice that we are not adding $(u,v)$ arc to $D$.
		
		\item For every pair of terminals $(s_i,t_i)$ $i \in [r]$, we do the following. First we add $2(k+1)$ vertices $\{s_i^j ; j \in [k+1]\} \cup \{t_i^j ; j \in [k+1]\}$ to $V(D)$. Next we add the set of arcs $\{(s_i^j, t_i^{j'}); j,j' \in [k+1] \}$ to $A(D)$. Then for each pair of vertices $x,y$ in $V(D)$ where either $x \in \{s_i^j ; j \in [k+1]\}$ and $y \in N^{+}(s_i) $ or  $y \in \{t_i^j ; j \in [k+1]\}$ and ~$x \in N^{-}(t_i) $, we add the arc $(x,y)$ to $A(D)$. Let $ V_{st} $ be the set of vertices that we add in this step.
	\end{enumerate}
	
	\noindent This completes the construction of the graph $D$. Observe that the number of vertices in $D$ is polynomial in $n$ and $k$, and the construction can be done in polynomial time.  Since $G$ is a  DAG, $D$ also forms a DAG. First, we make the following simple observation crucial for our reduction.
	
	\begin{observation}\label{obs_tran_hard}
		An arc $ e $ in $A(D) $ is transitive if and only if $ e \in  \{(s_i,t_i) | i \in [r]\} \cup  \{(s_i^j, t_i^{j'}); i \in [r]$ and $j,j' \in [k+1] \}$.
	\end{observation}
	
	\begin{clam}
		$(G,\mathcal{T},k)$  is a \yes-instance  if and only if $(D,k)$  is a \yes-instance.
	\end{clam}
	
	\begin{proof}
	In the forward direction, suppose that $(G,\mathcal{T},k)$ is a \yes-instance and $S \subseteq V(G)$ is a solution. We show that $ S $ is also a solution for $ (D,k) $, i.e., $D-S$ has no transitive arcs. Otherwise, let $e$ be a transitive arc in $D-S$. From Observation  \ref{obs_tran_hard}, $e$ must be from the set $   \{(s_i,t_i) | i \in [r]\} \cup  \{(s_i^j, t_i^{j'}); i \in [r], j,j' \in [k+1] \}$. If $ e= (s_i,t_i) $ for some $ i\in [r] $, then there is a path $P$ from $ s_i $ to $ t_i $ in $D-S-\{e\}$ with no internal vertices from $V_{st}$. If $ e= (s_i^j, t_i^{j'})$ for some $ i\in [r] $ and $ j,j' \in [k+1] $, then there is a path $ P'$ from $ s_i^j $ to $ t_i^{j'} $ containing the vertices both from $N^{-}(s_i)$ and $N^{+}(t_i)$, where $ P $ has no internal vertices from $ V_{st} $. Then we have a path from $ s_i$ to $ t_i $ in $D\setminus \Large\{V_{st} \setminus \{s_i^j,t_i^{j'} \}\Large\} $. But the existence of such paths $P$ and $P'$ from $s_i$ to $t_i$ in $G-S$ contradict the fact that $S$ is a solution to $(G,\mathcal{T},k)$. Hence $D-S$ has no transitive arcs.

	For the backward direction, suppose that $(D,k)$ is an \yes-instance and $Y \subseteq V(D)$ is a solution, i.e., $ D-Y $ has no transitive arcs. First, we modify $ Y $ to construct a solution $ Y' $ of size at most $|Y|$  for $(D,k)$ with the additional property that $ Y' $ has no vertices from $ V_E $. For the above purpose, we apply the following procedure: for any arc $e= (u,v) \in A(G)$ (i) if $ e_{uv} \in Y $ and any of $ u,v $, say $ v $ is not in $ Y $, then we discard $ e_{uv} $ and add $ v $, (ii) else if $ \{e_{uv},u,v\} \subseteq Y $,  then we discard $ e_{uv} $. We repeat this procedure until no vertices of $V_E $ are left in the solution. Let $ Y' $ be the set obtained following this procedure exhaustively. $Y'$ is also a solution, and $D-Y'$ has no transitive arcs. Let $ Z= Y'\setminus \{(s_i,t_i) | i \in [r]\} $. For each $ i\in [r] $, the graph $ G-Z $ can not have any path from $ s_i $ to $ t_i $ since there exists a pair of vertices $ p \in  \{s_i^j ; j \in [k+1]\} $ and $q \in \{t_i^j ; j \in [k+1]\}$ such that both $p,q  $ are not in $ Y' $. This, together with a path from $ s_i $ to $ t_i $ (and consequently the same path from $p$ to $q$), makes the arc $(p,q)$ transitive in $ D-Y' $, a contradiction. And hence $Z$ is a solution to the $(G,\mathcal{T},k)$ of size at most $k$.
	\end{proof}

This completes the proof of Theorem \ref{theo:W1-HARD-DAG}. \qed

\section{Kernelization on $\alpha$-bounded digraphs} \label{sec-kernel}

In this section, we design a kernel of  size $f(k,\ell, \alpha)$ for \lvd on $\alpha$ bounded in-tournaments. Due to  symmetry, an analogous result for out-tournaments can be obtained as well.

\subsection{Bounding the number of acyclic triangles}

In Lemma \ref{lem:triangle}, we prove that an in-tournament has no  transitive arcs  if and only if it does not contain any  acyclic triangles. This directly implies that any transitive arc of an in-tournament must be a part of some  acyclic triangle. In the first step towards our kernelization algorithm, we bound the number of acyclic triangles which gives us a set of arcs $A(\triangle)$ of bounded size that contains all possible transitive arcs. 
We now describe our procedures to obtain such a set $A(\triangle)$.

\begin{proc}
	\label{reduction:disjoint}
	 We  find a maximal set of vertex-disjoint acyclic triangles $\triangle_0$ in $D$ (greedily). If $|\triangle_0|>(\ell+k+1)$, then we conclude  $(D, k,\ell)$ to be a \no-instance.
\end{proc}

Note that for every acyclic triangle, $\{x,y,z\}$, the subgraph $D[{x,y,z}]$ includes a transitive arc. Therefore, if there are more than $(\ell+k+1)$ acyclic triangles in $D$ that are disjoint in terms of vertices, we must remove at least $(k+1)$ vertices to ensure that the remaining graph contains at most $\ell$ transitive arcs. Hence the correctness of procedure \ref{reduction:disjoint} follows.

For a vertex $v\in V(\triangle_0)$, $T_v=\{ \text{any triangle } \triangledown: V(\triangledown)\cap V(\triangle_0)=\{v\}\}$ and for an arc $e=(x,y)$, $T_e=\{ \text{any triangle } \triangledown: V(\triangledown)\cap V(\triangle_0)=\{x,y\}\}$ i.e., $T_v$ contains all those triangles that have exactly the vertex $v$ from $\triangle_0$ and $T_{(x,y)}$ contains all those triangles that have exactly the vertices $x$ and $y$ from $\triangle_0$.

\begin{proc}\label{reduction:one vertex}
     If for some  vertex $v$ in $V(\triangle_0)$, $|T_v| > (k+\ell+1) $, we arbitrarily delete a triangle from $T_v$.
\end{proc}

\begin{proc}\label{reduction:one edge}
If for some  arc $e$ in $D[V(\triangle_0)]$, $|T_e| > (k+\ell+1) $, we arbitrarily delete a triangle from $T_v$.
\end{proc}

Let $\triangle_1$ and $\triangle_2$ respectively denote the triangles intersecting at exactly one and two vertices with $V(\triangle_0)$.    
Following the applications of procedures \ref{reduction:one vertex} and \ref{reduction:one edge} we can show that $|\triangle_1|\leq (3(k+\ell+1)^2$ and $ |\triangle_2|\leq 9(k+\ell+1)^3)$. Let $\triangle = \triangle_0 \cup \triangle_1 \cup \triangle_2$. So we have $|\triangle| = \mathcal{O} ((k+\ell+1)^3)$.

\begin{clam}\label{claim:edge}
Let $S$ be a solution to $(D,k,\ell)$. Any transitive arc in $D-S$ belongs to $A(D[V(\triangle)])$.
\end{clam}
\begin{proof}
We prove the above claim using contradiction. Let $e=(x,y)$ be an arc which is transitive in $D-S$ but does not belong to 
$A(D[V(\triangle)])$. 
From Lemma \ref{lem:triangle} any transitive arc of $D-S$ (in-tournament) must be a part of some  acyclic triangle. So $e$ is also contained in some acyclic triangle of $D-S$, say  $\triangledown^*$ with $V(\triangledown^*)= \{x,y,z\}$. Clearly $\triangledown^*\notin \triangle_0$. Also from the maximality property of $\triangle_1$, $V(\triangledown^*)\cap V(\triangle_0)\neq \emptyset$. 
Hence $|V(\triangledown^*)\cap V(\triangle_0)|$ is either 1 or 2. If $x$ and $y$, both are not in $V(\triangle_0)$, then there are at least $(k+\ell+1)$ acyclic triangles exactly intersecting at $z$.  But  any solution $S$ of size at most $k$ must contain $z$, a contradiction to the assumption that $\triangledown^*$ is an acyclic triangle in $D-S$. Suppose exactly one of the vertices in $\{x,y\}$, say $x$ is not in $V(\triangle_0)$.  Since $(x,y)$ is a transitive arc, it must be part of some triangle $\triangledown^*$ in $D-S$ whose other vertex is $z$. We have the following 2 cases. $z$ is either in $V(\triangle_0)$ or not in $V(\triangle_0)$. In the latter case, there are at least $(k+\ell+1)$ triangles exactly intersecting at $y$. Hence $y$ must be in $S$, a contradiction. In the first case, $y$ and$z$ both are in $V(\triangle_0)$, but then there are at least $(k+\ell+1)$ triangles exactly containing both $y$ and $z$. But in this case, either $y$ or $z$ belong to the solution $S$, a contradiction.\end{proof}

\subsection{Designing the final kernel}

We use the following \Cref{thm1} crucially in the design of our kernel.

\begin{definition} [Cut-preserving Set~\cite{DBLP:conf/innovations/LochetLM0SZ20}]
	{\em For a digraph $D$, a positive integer $k$ and $x, y \in V (G)$, we say that $\mathcal{Z} \subseteq V(D)$ is a $k$-cut-preserving set for $(x, y)$ in $D$, if the following properties hold. 
	Let $L = V (D) \setminus \mathcal{Z}$. For any path $P$ from $x$ to $y$ in $G$, there exist paths $P_1, P_2, \ldots , P_q$ and a set of lists, $L_1, . . . , L_q$ each containing at  most $k$ paths with the following properties:
	
	\begin{enumerate}
		\item For every $i \in [q]$, $P_i$ is a subpath of $P$ from $s_i$ to $t_i$.
		\item The $P_i$s are internally disjoint and contain all vertices in $P \cap L$ as inner vertices.
		\item For every $i \in [q]$, the list $L_i$ is a set of $k$ vertex disjoint paths from $s_i$ to $t_i$ using only vertices of $\mathcal{Z}$.
		\item Replacing in $P$ each $P_i$ by one of the paths in $L_i$ yields a path of $\mathcal{Z}$ from $x$ to $y$.
	\end{enumerate}}
\end{definition}

Lochet et al.~\cite{DBLP:conf/innovations/LochetLM0SZ20} have shown the following results for an $\alpha$-bounded digraph.

\begin{proposition}[\cite{DBLP:conf/innovations/LochetLM0SZ20}]\label{thm1}
	Let $D$ be an $\alpha$-bounded acyclic digraph and $x, y \in V(D)$ such that any $(x, y)$-vertex-cut in $D$ has size at least $k+ 1$. Then one can, in polynomial-time, compute a $k$-cut-preserving set for 
	$(x, y)$ in $D$ of size at most ${(22k^5)}^{4^\alpha}$. Moreover in polynomial-time one can obtain $k+1$ vertex disjoint 
	paths from $x$ to~$y$ where each path has length at most $2\alpha +1$.
\end{proposition}

\begin{proc}[\cite{DBLP:conf/innovations/LochetLM0SZ20}]\label{reduction:cut set}
       For any arc $(x,y)\in A(D[V(\triangle)])$, we compute a $(k+1)$-cut-preserving set  $\mathcal{C}_{(x,y)}$ (by \Cref{thm1}). 
\end{proc}

For convenience, we use $A(\triangle)$ to denote $A(D[V(\triangle)])$ just for this section.  We claim
$V(\triangle)\cup \bigcup_{(x,y)\in A(\triangle)}\mathcal{C}_{(x,y)}\}$ is the desired kernel.

\begin{clam}\label{claim:kernel}
$D[X]$ where $X=V(\triangle)\cup \{\bigcup_{(x,y)\in A(\triangle)}\mathcal{C}_{(x,y)}\}$ is a kernel. 
\end{clam}

\begin{proof}
Let $(D,k,\ell)$  be a \yes-instance and $S$ be a solution. We denote the set of  transitive arcs in $D-S$ by $T_r(S)$. Let $e=(a,b)\notin T_r(S)$ be a transitive arc in $D[X-S]$. Then there is a path from $a$ to $b$ in $D[X-S]$. But this path is also present in $D-S$; hence, $e$ remains a transitive arc in $D-S$, which is a contradiction. This implies that $X\cap S$ is a set of at most $k$ vertices such that $D[X-S]$ has at most $\ell$ transitive arcs, i.e., $(D[X],k,\ell)$ is a \yes-instance.

Now let $(D[X],k,\ell)$  be a \yes-instance and 
$S$ is a corresponding solution. We denote the set of  transitive arcs in $D[X]-S$ by $T_r(X(S))$.
 Next, we show that $X$ is a  solution to $(D,k,\ell)$. From Claim \ref{claim:edge}, any transitive arc in any of the subgraphs of $D$ (and hence also in $D-S$) must be present in $A(\triangle)$ and hence in $X$.
Suppose $S$ is not a solution to $(D,k,\ell)$. Then there exists a transitive arc $e=(a,b)\in A(\triangle)\setminus T_r(X(S))$ in $D-S$. So there is path $P$ from $a$ to $b$ in $D-S-e$. If $V(P)\subseteq X$, then $e$ is also transitive in $D[X]-S$, a contradiction. 
Otherwise, let $P'$ be a maximal subpath of $P$ such that $V(P')\cap X=\emptyset$. Also, let $y$ and $z$ be the two adjacent vertices to the endpoints of $P'$ in $P$. Then in $\mathcal{C}_{(y,z)}$, there are at least $(k+1)$ vertex disjoint paths from $y$ to $z$. From the pigeon-hole principle, at least one of these paths is disjoint from $S$. Let $P''$ be such a path (which is completely contained in $D[\mathcal{C}_{(x,y)}]$ and hence also in $D[X]$). We can replace the subpath $P'$ by $P''$ in $P$. Notice that with every such replacement, the number of vertices from $P$ that is not in $X$ decreases. After exhaustively applying this procedure, we will have the path $P$ completely contained in $X$. But then, the arc $e=(a,b)$ is also transitive in $D[X]-S$ ($e\in T_r(X(S))$), a contradiction. This proves that no transitive arcs in $D-S$ are contained in $T_r(X(S))$. In other words, $S$ is a solution to  $(D,k,\ell)$ making it a \yes-instance.
\end{proof}
But $|X|\leq |V(\triangle)|+|A(\triangle)|. f(k+1,\alpha)\leq \mathcal{O} ((k+\ell+1)^3)+{(22(k+1)^5)}^{4^\alpha} $. So we have the following theorem.

\begin{theorem}
\lvd on in-tournaments, as well as out-tournaments, admits a kernel of size $\mathcal{O} ((k+\ell+1)^3)+{(22(k+1)^5)}^{4^\alpha}$ where $\alpha$ is the  size of a maximum independent set of the input digraph.
\end{theorem}

\section{FPT algorithms on in-tournaments/out-tournaments}\label{Sec-fpt}
In this section, we study \lvd for the special case when $\ell=0$ and the input graph is either an in-tournament or an out-tournament. 
First, we prove the following lemma, which is crucial in constructing our {\sf FPT} algorithm.

\begin{lemma}\label{lem:triangle}
	Let $D$ be an in-tournament (resp, out-tournament). Then $D$ has no transitive arcs  if and only if $D$ does not contain an acyclic triangle as an induced subgraph.
\end{lemma}

\begin{proof}
We prove the above lemma when $D$ is an in-tournament. A similar proof for when $D$ is an out-tournament can be constructed. First, we prove the forward direction using contradiction. Suppose $D$ has no transitive arcs but an induced acyclic triangle $\triangle$. W.L.O.G let $V(\triangle)=\{x,y,z\}$ and  $A(\triangle)=\{(x,y),(x,z),(z,y)\}$. But $(x,y)$ is a transitive arc in $D(\triangle)$, a contradiction.
	
	Suppose $D$ contains no induced acyclic triangles for the reverse direction. Let $e=(u,v)$ be a transitive  arc in $D$. Then, there exists a directed  path $P=\{u,a_1,\ldots,a_r,v\}$ in $D-e$. Since $u$ and $a_r$ are in-neighbors of $v$, by the definition of in-tournaments, there must be an arc $(u,a_r)$ or $(a_r, u)$ in $A(D)$. In the former case, we have an induced acyclic triangle with vertices $(u,v,a_r)$ with $(u,v)$ as the transitive arc. In the latter case, we have an induced acyclic triangle with the same set of vertices  and $(u,v)$ as the transitive arc. Thus we get contradictions for both cases.
\end{proof}

Notice that from Lemma \ref{lem:triangle}, \lvd on in-tournaments, as well as out-tournaments, reduces to the 3-{\sc Hitting Set} problem when $\ell=0$. In the  3-{\sc Hitting Set} problem, a universe $ U $, a family $\mathcal{F}$ of subsets of $ U $, and a non-negative integer $ k $  are given as inputs where the size of each set in $\mathcal{F}$ is at most 3. The goal is to check whether there exists  $U'\subseteq U $ of size at most $k$ that has a non-empty intersection with each set in $ \mathcal{F} $.

Given an instance $( D , k )$ of \zvd on $ D $ where $D$ is an in-tournament or an out-tournament with $ n $ vertices, one can create an equivalent instance $( U , \mathcal{F} , k' )$ of 3-{\sc Hitting Set} problem as follows:
\begin{itemize}
	\item[$\bullet$] $ U=V(D) $
	\item[$\bullet$] $ \mathcal{F} = \{ \{u,v,w\}| D[u,v,w] \text{ is an acyclic triangle}\} $
	\item[$\bullet$] $ k'=k $
\end{itemize}

As far as we know, the best-known algorithm for 3-{\sc Hitting Set} runs in time $ 2.0755^k n^{\mathcal{O}(1)} $ given by  Wahlström \cite{wahlstrom2007algorithms}. Hence we also get an algorithm for \zvd  on in-tournaments (resp, out-tournaments) running simultaneously. And the $\mathcal{O}(k^2)$ kernel for  3-{\sc Hitting Set} \cite{abu2010kernelization} can be adapted to design a kernel of similar size for \lvd, when $\ell=0$. Following the procedure,  as given in \cite{article}, one can improve the running time of the algorithm to $ 2^k n^{\mathcal{O}(1)}$ for  in-tournaments (resp, out-tournaments) and the  kernel size  to $\mathcal{O}(k)$ on local tournaments. Hence we have the following theorem. 

\begin{theorem}\label{theo-intour}
\zvd admits a $2^kn^{\mathcal{O}(1)}$ time {\sf FPT} algorithm and  $\mathcal{O}(k^2)$ kernel on in-tournaments (resp, out-tournaments). 
\end{theorem}

\section{Conclusion}

We investigate the {\sc $\ell$-RTVD} problem across various established classes of digraphs, including directed acyclic graphs (DAGs), planar DAGs, $\alpha$-bounded digraphs, tournaments, and their diverse generalizations such as in-tournaments, out-tournaments, local tournaments, and acyclic local tournaments. The problem can be solved in polynomial time for tournaments, $\alpha$-bounded digraphs, and acyclic local tournaments when $\ell$ is fixed. However, it is $\textsf{NP}$-Hard in planar DAGs with a maximum degree of 6. In terms of parameterized complexity, we find polynomial kernels for {\sc $\ell$-RTVD} on in-tournaments and out-tournaments when parametrized by $k+\ell$ for graphs with a bounded independence number. Nevertheless, the problem is fixed-parameter intractable for DAGs when only parameterized by $k$. Exploring {\sc $\ell$-RTVD} on DAGs using structural parameters is a potential avenue for further research.

\bibliography{main}

\end{document}